\documentclass{amsart}

\usepackage{epic}

\usepackage{latexsym}
\usepackage{amsmath}
\usepackage{epsfig}
\usepackage{amssymb}
\usepackage{enumerate}

\usepackage[mathscr]{euscript}

\newtheorem{theorem}{Theorem}[section]

\newtheorem{proposition}[theorem]{Proposition}

\begin{document}
\title{An Arrow-type result for inferring a species tree  from gene trees}
\author{Mike Steel}
\address{Allan Wilson Centre for Molecular Evolution and Ecology, University of Canterbury, Christchurch, New Zealand: Email: mike.steel@canterbury.ac.nz}
\begin{abstract}
The reconstruction of a central tendency `species tree' from a large number of conflicting gene trees is a central problem in systematic biology. Moreover, it becomes particularly problematic when taxon
coverage is patchy, so that not all taxa are present in every gene tree. Here, we list four  desirable properties that a method for estimating a species tree from gene trees should have.  We show that while
these can be achieved when taxon coverage is complete (by the Adams consensus method), they cannot all be satisfied in the more general setting of partial taxon coverage. 

\end{abstract}

\date{\today}

\maketitle

\section{Axioms for reconstructing a species tree from gene trees}
Reconstructing a rooted phylogenetic species tree from a sequence of gene trees (one for each genetic locus) can be viewed as a type of voting procedure. Each locus supports a gene tree,  and  tree reconstruction seeks to return a species tree that has overall highest support from the population of voters (trees).  As in social choice theory, where Arrow's theorem \cite{arr} has long played a prominent role,  it is  relevant to ask which properties can be satisfied; in  phylogenetics, a number of authors have considered these questions and shown that various combinations of axioms are impossible \cite{bar, bar2, day, mcm, mcm2, ste, tha}.  In this short note, we describe a further result based on a slightly different set of assumptions that are appropriate to settings where taxon coverage across loci is patchy \cite{san}, and we discuss  its implications briefly.

Formally, a {\em species tree estimator} is a  function  $\psi$ that assigns a rooted phylogenetic $X$-tree to any sequence $t_1,\ldots, t_k$ of trees at different loci, where $X$ is the set of taxa that occur in at least one tree.  Throughout this paper, all trees are rooted phylogenetic trees, and so can be thought of as a hierarchy (i.e. a collection of subsets of the nonempty leaf set $Y$, containing $Y$ and the singletons $\{y\}: y \in Y$, and satisfying the nesting property that any two sets are either disjoint or one is a subset of the other). 

If each tree $t_i$ has the same leaf set $X$ then $\psi$ constructs a {\em consensus tree},
while if the leaf set of the trees $t_1,\ldots, t_k$ are not all equal to $X$ (due to patchy taxon coverage across loci) then $\psi$ constructs a {\em supertree}.
A  tree reconstruction procedure is regarded as fully deterministic (e.g. in the case of ties, as with equally most parsimonious trees, one might take the strict consensus of the resulting trees). 

Here are four axioms.

\begin{itemize}
\item[{\bf (A1)}] (`Unrestricted domain') For any sequence (of any length) of rooted phylogenetic trees, $\psi$ returns a single phylogenetic tree (resolved or unresolved)  on some subset of the taxa mentioned by the input trees.
\item[{\bf (A2)}]   (`Unanimity') For any tree $t$, we have $\psi(t, t, \ldots, t) = t$.
\item [{\bf (A3)}] (A weak `Independence' condition).   
Suppose that $t_1, \ldots, t_k$ are rooted trees on leaf sets $X_1, \ldots, X_k$, respectively, and that $Y$ is a subset of $X=\cup_{i=1}^k X_i$.  Let
$t_i|Y$ denote the rooted phylogenetic tree that $t_i$ induces for $X_i \cap Y$.  
Then  the tree obtained by applying $\psi$ to $t_1|Y, \ldots, t_k|Y$ coincides with, or is refined by, the tree obtained from $\psi(t_1, \ldots, t_k)$ by considering just the taxa in $Y$. 
\item[{\bf (A4)}] (An `Irrelevance' axiom) 
Suppose that $t_1, \ldots, t_k$ are rooted trees on leaf sets that comprise subsets of $X$, and that $t_i$ consists of just one taxon that is present in the other trees. Then  $\psi(t_1, \ldots, t_k)$ is unchanged
if $\psi$ is applied to the sequence $t_1,\ldots, t_k$ with $t_i$ removed.
\end{itemize}

In words, (A3) states that if we add taxa and build a species tree, and we then consider how this tree describes the original set of taxa, it is either the same tree, or perhaps (due to the additional data) a more resolved one, but it is not
less resolved or inconsistent with the original tree.

Condition (A4) says that a rooted tree that just has one taxon in the set under consideration should not alter the tree returned by the method for that taxon set.  The idea here is that such a trivial tree carries no phylogenetic information, so it should not affect the outcome of the method.

Notice that (A3) can also be strengthened to the simpler statement, which we denote as (A3*) (this is listed as condition $I_6$ in \cite{bar2}):

\begin{itemize}
\item [{\bf (A3*)}] (An `Independence' condition).   
Suppose that $t_1, \ldots, t_k$ are rooted trees on leaf sets  that comprise subsets of $X$, and that $Y$ is a subset of $X$.  
Then  the tree obtained by applying $\psi$ to $t_1|Y, \ldots, t_k|Y$ is the tree obtained from $\psi(t_1, \ldots, t_k)$ by considering just the taxa in $Y$. 
\end{itemize}

In the consensus setting (where all the input trees have the same leaf set and so  (A4) holds vacuously)  an example of a method that satisfies (A1) and (A2) is to construct the strict consensus of all the input trees; however, this method fails (A3).  
Nevertheless, in this consensus setting, there is a method that satisfies all four properties (A1)--(A4) and even (A3*), namely the Adams consensus method \cite{ada}, as we now show.

\begin{proposition}
\label{prop1}
In the consensus setting, the Adams consensus method satisfies properties (A1)--(A4). Moreover,  (A3*) also holds.
\end{proposition}
\begin{proof}
First, as noted above, in the consensus setting, (A4) is vacuously satisfied. 
Conditions (A1) and (A2) are also clearly satisfied, so it suffices to establish (A3*); that is, for $\psi= Ad$, we have: $Ad(T_1|Y, \ldots, T_k|Y) = Ad(T_1, \ldots, T_k)|Y$.  
To see this, note that the maximal clusters of the tree $Ad(T_1|Y, \ldots, T_k|Y)$ are the nonempty intersections of maximal clusters of $T_1|Y,\ldots, T_k|Y$;  in other words,  
the nonempty intersections of the form  $(C_1 \cap Y) \cap (C_2 \cap Y) \cap  \cdots \cap  (C_k \cap Y)$, where $C_i$ is a maximal cluster of $T_i$.   However,
$$(C_1 \cap Y) \cap (C_2 \cap Y) \cap \cdots \cap (C_k \cap Y) = (C_1 \cap C_2 \cap \cdots \cap C_k) \cap Y,$$
and the right-hand side describes the maximal clusters of $Ad(T_1, \ldots, T_k)|Y$.  
Thus, $Ad(T_1|Y, \ldots, T_k|Y)$ and $Ad(T_1, \ldots, T_k)|Y$ have identical maximal clusters, and since the Adams consensus method proceeds recursively on the trees induced by these maximal clusters
it follows, by induction, that the consensus trees  $Ad(T_1|Y, \ldots, T_k|Y)$ and $Ad(T_1, \ldots, T_k)|Y$  have identical clusters.

\end{proof}

\section{An impossibility result in the supertree setting}

In the supertree setting, it also easy to find methods that simultaneously satisfy (A1), (A3) and (A4); a trivial example is the method that constructs
the star tree for all inputs.  

Satisfying (A1), (A2) and (A4) together is also fairly straightforward -- output the star tree unless, for some tree $t$, the input trees $(t_1,\ldots, t_k)$ have the property 
that $t_i = t$ for all $i$ in some nonempty subset $I$ of $\{1, \ldots, k\}$, and $t_j$ is a tree with just one leaf for all $j \in \{1, \ldots, k\} - I$; in which case we output the tree $t$.

There is also a supertree method that satisfies (A1), (A2) and (A3*) (and hence (A3)).  This is the following extension of the Adams consensus method.
Given $t_1, \ldots, t_k$  on the leaf sets $X_1, \ldots, X_k$ respectively, let $X' = \cap_{i=1}^k X_i$ and $X= \cup_{i=1}^k X_i$ denote, respectively, the set of leaves that are present in {\em every}  tree and the set of leaves that are present in {\em at least one} tree.

Construct $Ad(t_1|X' \ldots, t_k|X')$, the Adams consensus tree of $t_1|X' \ldots, t_k|X'$.  Now attach  each element $X-X'$ to the root of this tree by a separate pendant edge.
It is clear that this supertree method $\psi_{Ad}$ satisfies (A1) and (A2).  To see that $\psi_{Ad}$ also satisfies (A3*), suppose that $Y \subseteq X$. 
Then the intersection of the leaf sets of $t_1|Y, \ldots, t_k|Y$ is $W:=X' \cap Y$, and
since 
$$Ad(t_1|W \ldots, t_k|W) = Ad ((t_1|X')|Y, \ldots, (t_k|X')|Y) 
=Ad (t_1|X', \ldots, t_k|X')|Y,$$
where the second equality is from Proposition~\ref{prop1} (i.e. the Adams consensus method satisfies (A3*)), we have:
\begin{equation}
\label{eqcel}
Ad(t_1|W \ldots, t_k|W)  = Ad (t_1|X', \ldots, t_k|X')|Y.
\end{equation}

Now, $\psi_{Ad}(t_1, \ldots, t_k)$ is the tree obtained from $Ad(t_1|X', \ldots, t_k|X')$ by attaching each element of  $X-X'$  to the root of this tree by a separate pendant edge.
Thus $\psi_{Ad}(t_1, \ldots, t_k)|Y$ is the tree obtained from $Ad(t_1|X', \ldots, t_k|X')|Y$ by attaching each element of $Y-W$ to the root of this tree by a separate pendant edge.

On the other hand,   $\psi_{Ad}(t_1|Y \ldots, t_k|Y)$ is the tree obtained from $Ad(t_1|W \ldots, t_k|W)$ by attaching to the root of this tree each element of  $Y-W$ by a separate pendant edge.  Since $Ad(t_1|W \ldots, t_k|W) = Ad (t_1|X', \ldots, t_k|X')|Y$ (by (\ref{eqcel})),  we have  $\psi_{Ad}(t_1|Y \ldots, t_k|Y) = \psi_{Ad}(t_1, \ldots, t_k)|Y,$ as claimed.

\bigskip

This shows that the supertree method $\psi_{Ad}$ satisfies (A1), (A2) and (A3*).  However, what if we wish to include the apparently innocuous condition (A4) as well? In this case, even if we weaken (A3*) back to (A3), our main result shows that no such method can simultaneously accommodate these conditions.
Formally, we have:

\begin{proposition}
\label{main}
No tree reconstruction procedure exists that simultaneously satisfies axioms (A1), (A2), (A3) and (A4) on all inputs. 
\end{proposition}

\begin{proof}
We employ a proof by contradiction; that is,  by supposing there were a method method satisfying (A1)--(A4), we derive a contradiction. 

Our argument relies on the existence of a classic combinatorial object called a Steiner triple system (STS). This is a collection of 3-element subsets (called `blocks') from $\{1, 2,\ldots, n\}$ for which 
any two subsets intersection in exactly one point.  When an STS exists, it has exactly $b = \frac{n(n-1)}{6}$ blocks.  It is a  basic result in design theory (a branch of combinatorics \cite{van})
that an STS exists precisely when  the division of $n$ by 6  leaves a remainder of 1 or 3.  In particular, there exists an STS with $n=13$ ($=6 \times 2 +1$) and so with $b = 26$ blocks. 

Let us now suppose we have a method $\psi$ satisfying (A1)--(A4).   We take the taxon set as $X= \{1,2, \ldots, 13\}$ and we label the 26 blocks of the 
STS as $b_1, b_2, \ldots, b_{26}$.  For each block $b_i$,  let $t_{ij}$ (where $j=1,2,3$) denote the three
possible rooted binary trees we can construct that have the leaf set $b_i$.

Now, let $f: X \rightarrow \{1,2,3\}$ be a selection of one value of $j$ for each $i$, and consider the sequence $S_f$ of trees $t_{if(i)}$ of trees.  Each of these sequences of 26 trees
will comprise an input for $\psi$.

By (A1), $\psi(S_f)$ is a rooted phylogenetic tree, which  we will denote as $T_f$, on the  leaf set $X$ (or some subset of these leaves).

By (A3),  taking  the set $Y=b_k$ as our subset of taxa we  obtain:
\begin{equation}
\label{psik}
T_f|b_k \mbox{ equals or refines }\psi(S_f|b_k).
\end{equation}

Now, by (A2) (i.e. $\psi(t) = t$ for $t=t_{kf(k)}$) and by repeated applications of (A4) (it is here that we use the STS property that $|b_j \cap b_k|=1$ for all $j \neq k$), we have:

\begin{equation}
\label{psik2}
\psi(S_f|b_k) = \psi(S_{kf(k)}) = t_{kf(k)}.
\end{equation}

Combining (\ref{psik}) and (\ref{psik2}) (and noting that a rooted binary tree on three leaves cannot be further refined), we obtain:
\begin{equation}
T_f|b_k =  t_{kf(k)}.
\label{niceeq}
\end{equation}

Let $T'_f = T_f$ if the latter tree is binary; otherwise, let $T'_f$ denote any binary tree obtained from $T_f$ by resolving it arbitrarily.
Then:
\begin{equation} T'_f|b_k =  t_{kf(k)}.
\label{niceeq2}
\end{equation}

Notice that this implies that the leaf set of $T'_f$ must be all of $X$ .
Moreover, Eqn. (\ref{niceeq2}) holds for all $3^{26}$ possible choices for $f$.   This gives us $3^{26}$ rooted binary trees, each on the leaf set $X$ of size 13 (one tree for each choice of $f$).

At this point, we invoke a crucial arithmetic fact: $3^{26}$ is larger than the total number of rooted binary trees on 13 leaves, which is $(23)!! = 1 \times 3 \times \cdots \times 23$.
Thus, by the `pigeonhole principle' \cite{van}, at least two of the binary trees $T'_f$ and $T'_{f'}$ must be equal for some pair $f \neq f'$.  But, by (\ref{niceeq2}),  this implies that $t_{kf(k)}= t_{kf'(k)}$ for
all $k$, and so $f=f'$. This contradiction establishes that the initial assumption of the existence of a method satisfying (A1)--(A4) is not possible.

\end{proof}

\section{Discussion}

Suppose we have a fixed set $S$ of species. 
It is clear that (even in the consensus setting) any method for building a species tree from gene trees should allow the tree to change as more loci are sequenced and the gene trees for these loci are  included in the analysis (since the gene trees at later loci may, for example
favour a different species tree). 

But suppose we fix the set of available loci, and instead try to build a tree by adding taxa.  Thus we may construct a tree for some of the taxa and then sequentially try to attach each additional taxon in an optimal place in this tree. On occasions, an additional taxon may even allow us to resolve
the tree a bit better, but we do not wish to go back and rearrange the tree we obtained at an earlier stage of the process.  Proposition~\ref{main} assures us that there is  no method that can guarantee to achieve this goal in general while also satisfying the clearly desirable properties (A1) and (A2).

\section{Acknowledgments}

I thank F.R. McMorris for several helpful comments, and  advice on the Adams consensus literature, and Joel Velasco for discussions on the possible relevance of this result.

{}

\begin{thebibliography}{}
\bibitem{ada} Adams, E.N. III (1986). N-trees as nestings: complexity, similarity and consensus. Journal of Classification 3, 299--317.
\bibitem{arr}  Arrow, K.J., (1950). A difficulty in the concept of social welfare, Journal of Political Economy 58(4):  328--346.
\bibitem{bar} Barth{\'e}lemy, J.P., McMorris, F.R. and Powers, R.C. (1991). Independence conditions for consensus $n$-trees revisted. Applied Mathematics Letters 4: 43--46.
\bibitem{bar2} Barth{\'e}lemy, J.P., McMorris, F.R. and Powers, R.C. (1995).  Stability conditions for consensus functions defined on $n$-trees. Mathematical and Computer Modelling, 22: 79--87.
\bibitem{day} Day, W.H.E. and McMorris, F.R. (2003) Axiomatic consensus theory in group choice and biomathematics. SIAM, Philadephia. 
\bibitem{mcm} McMorris, F.R. (1985).  Axioms for consensus functions on undirected phylogenetic trees. Mathematical Biosciences 74: 17-21.
\bibitem{mcm2} McMorris, F.R. and Powers, R.C. (1993).  Consensus functions on trees that satisfy an independence axiom. Discrete Applied Mathematics. 47: 47--55. 
\bibitem{san} Sanderson, M.J., McMahon, M.M. and Steel, M. (2010). Phylogenomics with incomplete taxon coverage: the limits to inference. BMC Evolutionary Biology 10: 155. 
\bibitem{sem} Semple, C. and Steel, M. (2003) Phylogenetics. Oxford University Press. 
\bibitem{ste} Steel, M., B{\"o}cker, S., and Dress, A.W.M. (2000). Simple but fundamental limits for supertree and consensus tree methods. Systematic Biology 49(2): 363--368.
\bibitem{tha} Thatte, B. (2007).  A correct proof of the McMorris--Powers' theorem on the consensus of phylogenies. Discrete Applied Mathematics 155(3): 423--427. 
\bibitem{van}  van Lint, J.H. and Wilson, R.M. (2001). A course in combinatorics  (2nd ed.) Cambridge University Press.
\end{thebibliography}
\end{document}